\documentclass[conference]{IEEEtran}
\IEEEoverridecommandlockouts

\usepackage{cite}
\usepackage{amsmath,amssymb,amsfonts,amsthm}
\usepackage{graphicx}
\newtheorem{proposition}{Proposition}
\newtheorem{corollary}{Corollary}
\usepackage{textcomp}
\usepackage{xcolor}
\usepackage{xurl}
\usepackage{booktabs}
\usepackage{multirow}
\usepackage{array}
\usepackage{algorithm}
\usepackage{algpseudocode}
\usepackage{tikz}
\usetikzlibrary{arrows.meta,positioning,fit,backgrounds,shapes.geometric,calc}

\newcommand{\fittab}[1]{%
  \resizebox{\ifdim\width>\columnwidth\columnwidth\else\width\fi}{!}{#1}}
\newcommand{\tabnote}[1]{%
  \par\vspace{2pt}\parbox{\columnwidth}{\scriptsize\raggedright #1}}
\def\BibTeX{{\rm B\kern-.05em{\sc i\kern-.025em b}\kern-.08em
    T\kern-.1667em\lower.7ex\hbox{E}\kern-.125emX}}

\begin{document}
\bstctlcite{IEEEexample:BSTcontrol}

\title{Cost-Aware Post-Hoc Deferral Under Calibration and Shift:\\
An Environmental AI Case Study}

\author{%
\IEEEauthorblockN{Haoran Yu}
\IEEEauthorblockA{University of Florida\\
Gainesville, FL, USA\\
haoranyu889@gmail.com}
\and
\IEEEauthorblockN{Lifei Liu}
\IEEEauthorblockA{Wichita State University\\
Wichita, KS, USA\\
lliu.lifei@gmail.com}
\and
\IEEEauthorblockN{Danping Zhang}
\IEEEauthorblockA{Nanchang Hangkong University\\
Nanchang, China\\
zhangdanping@nchu.edu.cn}%
}

\maketitle

% -------------------------------------------------------
% Figure 1: environmental stakes, governance pipeline, and deployment action.
% -------------------------------------------------------
\begin{figure*}[t]
\centering
\includegraphics[width=0.97\textwidth]{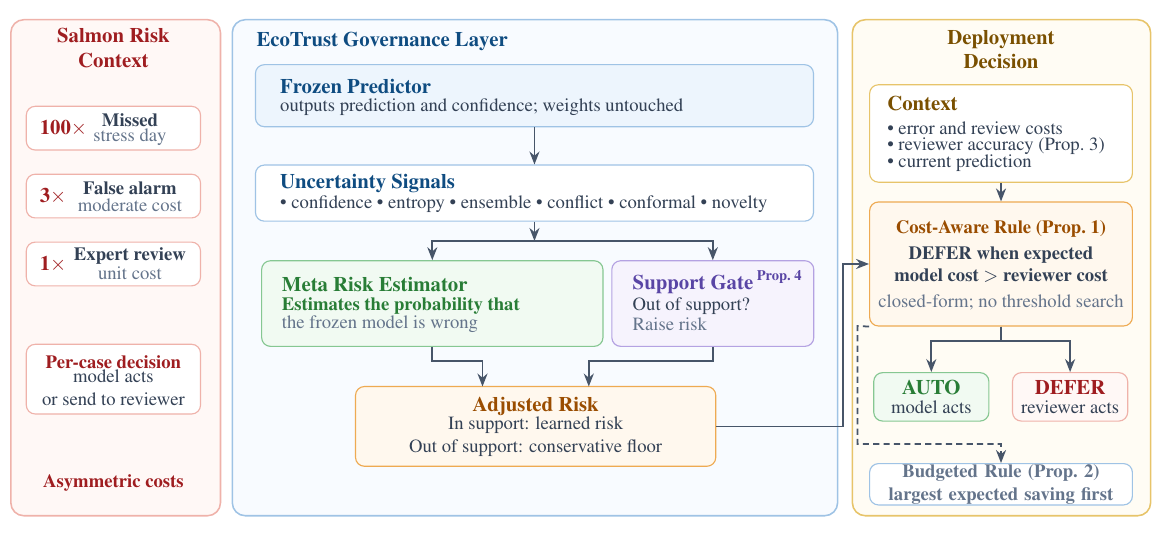}
\caption{EcoTrust studies per-case deferral for a frozen predictor under asymmetric costs. Six uncertainty signals estimate model error; the decision layer combines adjusted risk with error costs, review cost, and reviewer accuracy. Outside calibration support, a pre-specified risk floor invokes a conservative fallback. An optional budget ranks cases by expected saving.}
\label{fig:architecture}
\end{figure*}

% -------------------------------------------------------
\begin{abstract}
Choosing a deferral policy for a frozen classifier requires more than ranking uncertain cases: confidence may be miscalibrated, errors have unequal costs, reviewers can err, and deployment data can leave calibration support. We study these interactions through \textbf{EcoTrust}, a post-hoc framework that compares automatic action with review using a six-group error-risk estimator, class-asymmetric costs, reviewer accuracy, and an optional support gate. On a Columbia River thermal-stress testbed, the learned estimator improves error-ranking area under the receiver operating characteristic curve from 0.869 to 0.889, but Chow's confidence rule has lower in-distribution cost (0.416 versus 0.567 per day). Across 12 off-the-shelf backends, learned risk and a calibration-matched, class-aware confidence estimator each beat raw Chow on six; a paired year-block bootstrap does not resolve their mean cost difference. In transfer to ten river stations, the gate flags every case and becomes an always-review fallback, attaining the lowest cost on eight stations only when review is perfect and unconstrained. These results characterize decision boundaries on one controlled task: richer risk signals do not reliably improve on calibrated confidence, and detected extrapolation does not imply transferable case-level ranking.
\end{abstract}


% -------------------------------------------------------
\section{Introduction}
\label{sec:intro}

Consider a migration-season day at the Columbia River at The Dalles. Water above $18^\circ$C can impair Chinook, Coho, and Chum salmon~\cite{richter2005temperature}, but a frozen classifier can only provide a probability and a binary prediction. The operational question is who should act on that prediction. In our illustrative primary scenario, a missed stress day costs 100 review units, a false alarm costs three, and expert review costs one; reviewing every day is nevertheless incompatible with a limited workload. The same policy must also specify what happens when it is applied to a river outside its calibration support.

Reject-option and deferral research provides several parts of this decision. Chow's rule~\cite{chow1970} and selective classification~\cite{selective_class} route low-confidence predictions away from automatic action, while learning to defer (L2D) trains a rejector around expert behavior~\cite{learn_to_defer,mozannar2020,verma2022}. Post-hoc L2D estimates the error of a fixed model~\cite{narasimhan2022posthoc}. Closely related work has asked when confidence-based cascade deferral is sufficient~\cite{jitkrittum2023confidence}, evaluated post-hoc confidence estimators for selective classification~\cite{cattelan2024broken}, and connected post-hoc L2D to density-ratio estimation and Chow-style rules~\cite{soen2026density}. The unresolved issue here is therefore narrower: whether richer uncertainty signals improve deployment cost over a calibration-matched, class-aware confidence estimator in a human-review setting with asymmetric errors, imperfect review, quotas, and detected extrapolation.

Our goal is an empirical characterization, not a new Bayes decision rule. We ask when a separately learned error-risk model changes the action selected around a frozen predictor, when confidence is sufficient, and when an out-of-support detector should trigger a predetermined fallback. The Columbia River task supplies one controlled setting in which the same predictor and consequence model can be held fixed while calibration, reviewer quality, review rate, and deployment site vary.

This setting raises three evaluation challenges. First, raw confidence and learned risk are not comparable unless both receive the same calibration data and model capacity; otherwise supervised recalibration can be mistaken for value from additional signals. Second, ranking quality alone does not determine deployment value because false negatives, false positives, review cost, reviewer error, and forced review rates interact. Third, a shift detector can identify unsupported inputs without knowing whether the downstream risk ranking remains useful, so a safe fallback must be distinguished from transfer competence.

\textbf{EcoTrust} makes these comparisons explicit (Fig.~\ref{fig:architecture}). A cross-fitted meta risk estimator combines six signal groups, and a calibration-matched confidence estimator isolates the effect of using additional signals. A cost-aware rule maps either risk estimate to automatic action or review under class-asymmetric costs and reviewer accuracy. A support gate raises risk outside the training support, while fixed-rate experiments separate fallback behavior from ranking quality. Algorithm~\ref{alg:governance} summarizes the resulting deployment procedure.

\begin{itemize}
    \item \textbf{A cost-aware post-hoc formulation} (Secs.~\ref{sec:problem}--\ref{sec:framework}) that states the assumptions linking class-asymmetric errors, reviewer accuracy, review quotas, and a detected-extrapolation fallback.
    \item \textbf{A calibration-matched characterization} (Secs.~\ref{sec:main}--\ref{sec:reviewer}) showing that better error ranking need not lower cost and that the six-signal estimator does not reliably beat a class-aware confidence recalibrator across the studied backend bank.
    \item \textbf{Deployment boundaries under shift} (Sec.~\ref{sec:ood}) that separate an always-review fallback from transferable risk ranking using ten stations and a reviewer-accuracy by forced-review-rate grid.
\end{itemize}

% -------------------------------------------------------
\section{Related Work}
\label{sec:related}

\subsection{Human Deferral, Rejection, and Selective Prediction}
Chow~\cite{chow1970} established the reject option, and selective classification~\cite{selective_class} traded coverage against risk. L2D jointly trained a classifier and rejector from expert decisions~\cite{learn_to_defer}; later work derived consistent softmax~\cite{mozannar2020} and calibrated one-vs-all surrogates~\cite{verma2022}. We evaluate frozen-backend restrictions of the latter two losses, not reproductions of joint classifier--rejector training.

The most relevant post-hoc work spans three neighboring settings. Narasimhan et al.~\cite{narasimhan2022posthoc} estimated a fixed model's error and compared it with expert cost. Jitkrittum et al.~\cite{jitkrittum2023confidence} characterized when confidence suffices in model cascades, including under distribution shift, while Cattelan and Silva~\cite{cattelan2024broken} compared post-hoc confidence estimators for selective classification. Soen et al.~\cite{soen2026density} formulated post-hoc L2D through density-ratio estimation. DeCCaF incorporated class-dependent errors, multiple experts, and workload constraints~\cite{alves2024deccaf}; probabilistic and conformal variants addressed missing annotations, workload allocation, and rejector uncertainty~\cite{nguyen2025probdefer,fang2026conformaldefer}. EcoTrust does not replace these methods. It isolates a narrower deployment comparison: raw confidence, class-aware calibrated confidence, and multi-signal error risk for one reviewer under asymmetric consequences and detected extrapolation. Table~\ref{tab:positioning} summarizes this scope.

\begin{table}[t]
\caption{Deployment settings of selected deferral frameworks.}
\label{tab:positioning}
\centering
\scriptsize
\setlength{\tabcolsep}{1.5pt}
\fittab{%
\begin{tabular}{lccccc}
\toprule
\textbf{Framework} & \textbf{Backend} & \textbf{Costs} & \textbf{Reviewer} & \textbf{Capacity} & \textbf{Out-of-support} \\
 & & & & & \textbf{Fallback} \\
\midrule
Chow~\cite{chow1970} & Fixed & Reject & N/A & No & No \\
Post-hoc L2D~\cite{narasimhan2022posthoc} & Fixed & Expert & Cost model & No & No \\
DeCCaF~\cite{alves2024deccaf} & AI + experts & Class/inst. & Expertise model & Yes & No \\
\textbf{EcoTrust} & Fixed & FN/FP/review & Accuracy input & Optional & Yes \\
\bottomrule
\end{tabular}}
\tabnote{``Explicit fallback'' means a pre-specified action after detected extrapolation; it does not imply general robustness to distribution shift.}
\end{table}

\subsection{Uncertainty Quantification and Environmental AI}
Conformal prediction provides distribution-free marginal coverage under exchangeability~\cite{conformal_gentle}, while ensemble disagreement~\cite{ensemble_uncertainty} and entropy provide uncertainty signals. Environmental studies have modeled salmon productivity~\cite{jones2020chinook} and water temperature~\cite{noaa_stream_temp,streamtemp_gnn}. Xu et al.~\cite{salmon_foundation} combined foundation models with expert input for salmon fisheries; our study instead holds the environmental predictor fixed and evaluates the downstream deferral decision.

% -------------------------------------------------------
\section{Problem Formulation}
\label{sec:problem}

Let $x_t$ denote the environmental observation on day $t$: NOAA weather features (TMAX, TMIN, PRCP, AWND) and USGS hydrological features (discharge, gage height, derived flow statistics). The binary target is
\begin{equation*}
y_t = \mathbf{1}[\text{water\_temp}_t > 18^\circ\text{C}],
\end{equation*}
the study's operational thermal-stress threshold for Chinook salmon. Water temperature is measured independently by USGS and excluded from the features; month and day-of-year are excluded too, so the models use environmental measurements rather than calendar patterns.

A \emph{frozen predictor} $\mathcal{M}$ emits probability $\hat{p}(x)$ and decision $\hat{y}(x)=\mathbf{1}[\hat{p}(x)>\tfrac12]$. It remains fixed throughout governance training and evaluation.

A \emph{governance policy} $\pi:\mathcal{X}\to\{\mathrm{auto},\mathrm{defer}\}$ decides, per case, who acts. Its quality is the expected \emph{deployment cost}
\begin{equation}
\label{eq:cost}
L(\pi) \;=\; \mathbb{E}\bigl[\,c_{FN}\!\cdot\!\mathrm{FN} \;+\; c_{FP}\!\cdot\!\mathrm{FP} \;+\; c_h\!\cdot\!\mathbf{1}[\pi(x)=\mathrm{defer}]\,\bigr],
\end{equation}
where FN and FP count errors in the \emph{final} decision. Thus, an incorrect review receives the same consequence cost as an incorrect automatic decision. The deployment parameters are $(c_{FN},c_{FP},c_h)$, reviewer accuracy $a=P(\text{expert correct}\mid\mathrm{defer})$, and, optionally, review budget $B$. This formulation operationalizes our comparison: error coverage, the fraction of model errors escalated, is cost-blind because it weighs a miss and a false alarm equally. We report it for comparability but optimize~\eqref{eq:cost}.

% -------------------------------------------------------
\section{EcoTrust Deployment Framework}
\label{sec:framework}

\subsection{Stage 1: The Meta Risk Estimator}

Let $s(x)$ collect the six signal groups in Fig.~\ref{fig:architecture}: confidence $|\hat{p}-\tfrac12|$, entropy $H(\hat{p})$, ensemble standard deviation $\sigma_{\text{ens}}$, agent conflict $\mathbf{1}[|\hat{p}_W-\hat{p}_H|>\theta_c]$, two conformal quantities (predicted-label nonconformity and whether the $1-\alpha$ set contains both classes), and distance $\log(1+d_{\text{Mah}}(x,\mathcal{D}_{\text{tr}}))$.

The meta risk estimator is a small model
\begin{equation}
\label{eq:risk}
g_\phi:\; s(x)\;\longmapsto\; \hat{r}(x)\;=\;\hat{P}\bigl(\hat{y}(x)\neq y \,\big|\, s(x)\bigr)
\end{equation}
Equation~\eqref{eq:risk} is fitted by maximum likelihood on held-out data, with binary target $\mathbf{1}[\hat{y}\neq y]$. The target uses no recorded review decision: ordinary supervised calibration data suffice, and EcoTrust does not fit an expert-specific performance model.

\textbf{Cross-fitting.} The 488-day calibration block contains too few errors to fit $g_\phi$ reliably. We therefore use five expanding, year-blocked folds. For each validation block, the ensemble, conformal model, and support model are refitted using strictly earlier dates; early years serve only as warm-up until both classes are present. Combining 707 out-of-fold rows with the calibration block gives 1{,}195 meta-training rows with an 8.2\% error rate. The deployed $\mathcal{M}$ remains trained only on the designated training block.

\textbf{Calibration and ranking.} The decision rule compares $\hat r$ with a cost-derived threshold. A miscalibrated estimate can therefore produce an incorrect review rate even when its ranking is accurate. We use cross-fitted Platt scaling for $g_\phi$.

\subsection{Stage 2: The Cost-Aware Decision Rule}

Write $c_{\text{err}}(\hat{y}) = c_{FP}$ if $\hat{y}{=}1$ and $c_{FN}$ if $\hat{y}{=}0$: this is the cost incurred \emph{if} the automatic decision turns out to be wrong, and it is determined by the predicted label, since a wrong positive is a false alarm and a wrong negative is a miss. The two available actions then have conditional expected costs
\begin{align}
\mathbb{E}[\text{cost}\mid\mathrm{auto}] &= \hat{r}(x)\cdot c_{\text{err}}(\hat{y}(x)), \label{eq:auto}\\
\mathbb{E}[\text{cost}\mid\mathrm{defer}] &= \underbrace{c_h + (1-a)\bigl(\hat{p}\,c_{FN} + (1-\hat{p})\,c_{FP}\bigr)}_{\textstyle =:\;C_{\mathrm{rev}}(x)}. \label{eq:defer}
\end{align}
Equation~\eqref{eq:defer} includes both review cost and the expected consequence of reviewer error. With accuracy $a$, the reviewer errs with probability $1-a$; under calibrated $\hat p$, that error is a miss with probability $\hat p$ and a false alarm otherwise. The policy selects the lower-cost action.

The comparison below is standard in reject-option and post-hoc deferral~\cite{chow1970,narasimhan2022posthoc}. We state it to fix the asymmetric cost semantics used in every experiment.

\begin{proposition}[Conditional decision rule]
\label{prop:bayes}
Fix any predictor $\mathcal{M}$ and let $r(x)=P(\hat{y}(x)\neq y\mid x)$ be the true error probability. Assume the reviewer errs independently of the label and $\hat{p}$ is calibrated, so $C_{\mathrm{rev}}(x)$ in~\eqref{eq:defer} is the true deferral cost. Among \emph{all} measurable policies $\pi:\mathcal{X}\to\{\mathrm{auto},\mathrm{defer}\}$, the cost~\eqref{eq:cost} is minimized by
\begin{equation}
\label{eq:rule}
\pi^\star(x)=\mathrm{defer} \quad\iff\quad r(x)\,c_{\text{err}}(\hat{y}(x)) \;>\; C_{\mathrm{rev}}(x).
\end{equation}
\end{proposition}

\begin{proof}
Condition on $x$. The policy chooses between two actions whose conditional expected costs are exactly~\eqref{eq:auto} and~\eqref{eq:defer}, and $L(\pi)=\mathbb{E}_x\bigl[\mathbb{E}[\text{cost}\mid \pi(x),x]\bigr]$ is an integral of a pointwise choice. An integral of pointwise minima is minimized by minimizing pointwise, and the pointwise minimizer defers precisely when~\eqref{eq:rule} holds.
\end{proof}

We do not claim the pointwise Bayes comparison as new. Its role is to make the deployment assumptions explicit and reduce the learned component to estimating $r(x)$. The deployed rule substitutes $\hat r$ and $\hat p$, so the result applies only to the extent that both are accurate and calibrated.

\begin{corollary}[Closed-form asymmetric thresholds]
\label{cor:thresholds}
With a perfect reviewer ($a{=}1$), rule~\eqref{eq:rule} becomes a threshold on the error probability that depends on the \emph{predicted class}:
\begin{equation}
\label{eq:thresholds}
\mathrm{defer} \iff
\begin{cases}
r(x) > c_h/c_{FN} & \text{if } \hat{y}(x)=0,\\[1pt]
r(x) > c_h/c_{FP} & \text{if } \hat{y}(x)=1.
\end{cases}
\end{equation}
\end{corollary}

At $c_{FN}{:}c_{FP}{:}c_h = 100{:}3{:}1$, the thresholds in~\eqref{eq:thresholds} are $0.01$ and $0.333$. The system therefore tolerates much less estimated risk for a predicted-safe day. A single class-independent threshold cannot express this asymmetry.

\begin{proposition}[Convex cost--budget frontier]
\label{prop:budget}
Under a review budget of $k$ cases, the cost-minimizing policy defers the $k$ cases with the largest expected saving $\Delta(x)=r(x)c_{\text{err}}(\hat{y})-C_{\mathrm{rev}}(x)$, truncated at $\Delta>0$. The resulting optimal cost $L^\star(k)$ is non-increasing and convex in $k$, and the escalated sets are nested, so error coverage is non-decreasing in $k$.
\end{proposition}

\begin{proof}
Deferring case $i$ changes the cost by exactly $-\Delta(x_i)$, independently of which other cases are deferred, so the budgeted problem is a selection of the $k$ largest $\Delta$'s. Writing $\Delta_{(1)}\geq\Delta_{(2)}\geq\cdots$ for the sorted savings, $L^\star(k)=L^\star(0)-\sum_{i\leq k}\Delta_{(i)}$. Its increments $-\Delta_{(k)}$ are non-decreasing in $k$, which is convexity; and $\Delta_{(k)}>0$ up to the truncation point gives monotonicity. The top-$k$ sets are nested by construction.
\end{proof}

Convexity formalizes diminishing returns on review capacity. We verify it with the oracle at-most-$k$ diagnostic; Fig.~\ref{fig:operating}(b) separately compares deployable rankings at fixed review rates.

\begin{proposition}[Breakeven reviewer accuracy]
\label{prop:breakeven}
Let $S$ be any fixed escalation set. Deferring $S$ raises expected \emph{accuracy} above the fully automatic system if and only if
\begin{equation}
\label{eq:breakeven}
a \;>\; 1 - \mathbb{E}[\,r(x)\mid x\in S\,] \;=\; \text{the model's accuracy on } S .
\end{equation}
\end{proposition}

\begin{proof}
Outside $S$ nothing changes. On $S$ the model is correct with probability $1-\mathbb{E}[r\mid S]$ and the reviewer with probability $a$, so the accuracy change is $\tfrac{|S|}{n}\bigl(a-(1-\mathbb{E}[r\mid S])\bigr)$, which is positive exactly under the stated condition.
\end{proof}

The breakeven in~\eqref{eq:breakeven} is a computable property of the escalation set. Section~\ref{sec:reviewer} verifies that the measured crossing lies between reviewer accuracies 0.90 and 0.95, around the predicted 0.9305.

\subsection{Support Gate and Detected Extrapolation}

Every property above is conditional on $\hat r$ estimating $r$. This condition may fail under distribution shift: $g_\phi$ was fitted on in-distribution calibration data and can extrapolate toward low risk on a new river (Sec.~\ref{sec:ood}).

We therefore give the policy an explicit validity condition. Let $d(x)$ be the squared Mahalanobis distance to the training feature distribution and $d_q$ its $(1-\epsilon)$ training quantile ($\epsilon{=}0.05$). Define the \emph{deployed risk}
\begin{equation}
\label{eq:gate}
\tilde{r}(x) = \begin{cases}
\hat{r}(x), & d(x)\leq d_q \quad\text{(in support)},\\
\max\bigl(\hat{r}(x),\ \tfrac12\bigr), & d(x) > d_q \quad\text{(extrapolating)}.
\end{cases}
\end{equation}
The floor in~\eqref{eq:gate} corresponds to maximum uncertainty for a binary error indicator. The gate uses neither labels nor target-domain fitting and can therefore be applied from the first deployment day at a new site.

\begin{proposition}[Fallback cost on out-of-support cases]
\label{prop:gate}
Substituting $\tilde{r}$ for $\hat{r}$ in rule~\eqref{eq:rule} yields two properties on out-of-support cases. \emph{(i)} With a perfect reviewer ($a{=}1$), the policy defers whenever $c_h/c_{\text{err}}(\hat{y}) < \tfrac12$, so an error costs more than twice a review. \emph{(ii)} Conditional on deferral, the expected cost is at most $c_h + (1-a)\max(c_{FN},c_{FP})$, independent of $\mathcal{M}$'s error, and equals $c_h$ when $a{=}1$.
\end{proposition}

\begin{proof}
\emph{(i)} On an out-of-support case $\tilde{r}\geq\tfrac12$, so the left side of~\eqref{eq:rule} is at least $c_{\text{err}}(\hat{y})/2$; with $a{=}1$ this exceeds $c_h$ precisely when $c_h/c_{\text{err}}(\hat{y})<\tfrac12$. \emph{(ii)} A deferred case incurs $c_h$ plus the reviewer's expected error cost, which is bounded by $(1-a)\max(c_{FN},c_{FP})$ and does not involve $\mathcal{M}$ at all.
\end{proof}

On flagged cases, the model's possibly corrupted confidence no longer decides whether the fallback is invoked. This is a conditional fail-safe, not a robustness guarantee: concept shift without covariate movement and undetected covariate shift remain outside its scope.

\begin{algorithm}[t]
\caption{EcoTrust deployment rule}
\label{alg:governance}
\begin{algorithmic}[1]
\Statex \textbf{Stage 1: Fit risk and support models} ($\mathcal{M}$ remains fixed)
\State Cross-fit $\mathcal{M}$'s recipe on $\mathcal{D}_{\text{train}}$ for out-of-fold signals $s(x_i)$ and errors $e_i{=}\mathbf{1}[\hat{y}_i{\neq}y_i]$; append $\mathcal{D}_{\text{cal}}$
\State Fit meta risk estimator $g_\phi:\ s(x)\mapsto \hat{r}$ on $\{(s(x_i),e_i)\}$ \Comment{Platt-calibrated}
\State Fit support gate: $d_q \leftarrow (1{-}\epsilon)$ quantile of $d_{\text{Mah}}(x,\mathcal{D}_{\text{train}})$
\Statex \textbf{Stage 2: Deploy} (per case; no search, no tuning)
\State $\hat{r}\leftarrow g_\phi(s(x_t))$;\quad $\tilde{r}\leftarrow \hat{r}$ if $d(x_t)\!\leq\! d_q$ else $\max(\hat{r},\tfrac12)$
\If{$\tilde{r}\cdot c_{\text{err}}(\hat{y}_t) > C_{\mathrm{rev}}(x_t)$}
  \State $e_t\!\leftarrow\!1$; route to expert; $d_t\!\leftarrow$ expert decision
\Else\ \ $e_t\!\leftarrow\!0$; \ $d_t\!\leftarrow\!\hat{y}_t$
\EndIf
\Statex \textbf{Budgeted variant:} defer the top-$B$ fraction by expected saving $\Delta(x)$ (Prop.~\ref{prop:budget})
\end{algorithmic}
\end{algorithm}

% -------------------------------------------------------
\section{Experimental Setup}
\label{sec:setup}

\textbf{Data.} We combine daily National Oceanic and Atmospheric Administration weather from Climate Data Online at Portland International Airport (USW00024157) with United States Geological Survey National Water Information System hydrology at the Columbia River at The Dalles (14105700) for 1996--2024. All inputs are observed environmental measurements. Days without water temperature are dropped because their labels are undefined. Evaluation covers the August--November migration season.

\textbf{Splits.} The split is strictly temporal: train $<$2011 ($n{=}1{,}086$, 69.8\% positive), calibration 2011--2014 ($n{=}488$, 54.3\%), and test 2015+ ($n{=}1{,}208$, 55.4\%). The random-forest, XGBoost, and gradient-boosting ensemble is fitted only on the training block and reaches 92.7\% test accuracy. Every learned auxiliary signal in cross-fitting uses earlier dates only.

\textbf{Costs.} The primary scenario uses $c_{FN}{:}c_{FP}{:}c_h = 100{:}3{:}1$, assigning a miss a much greater cost than a false alarm and taking one review as the unit cost. This ratio is illustrative rather than an empirically measured management cost. Section~\ref{sec:costsweep} therefore sweeps $c_{FN}$ from 1 to 500. A daily cost of 1.0 equals one review with no decision error.

\textbf{Meta risk estimator.} Unless noted, $g_\phi$ is an $\ell_2$-regularized logistic regression on the standardized, median-imputed six-group signal bank, wrapped in cross-fitted Platt scaling (\texttt{CalibratedClassifierCV}, sigmoid, five inner folds) so each risk value comes from a model that did not see that case (Sec.~\ref{sec:framework}). No hyperparameter is tuned on test; the regularization is the scikit-learn default and folds are fixed a priori. The identical estimator is used for every backend in Sec.~\ref{sec:calib}; only the frozen predictor changes. This logistic estimator is the reported EcoTrust (C1); Fig.~\ref{fig:main} also plots a gradient-boosted variant (C2) and a no-gate variant (C3), which cluster at the same point.

\textbf{Baselines.} All use the same frozen predictor: A0 is no governance, A1 a cost-sensitive automatic threshold, B0 random deferral, B1--B2 earlier hand-crafted heuristics, B3 selective classification, B4 Chow's rule, and B5 a frozen-backend raw-feature rejector adapted from Cortes et al.~\cite{cortes2016}. B6 and B7 are post-hoc frozen-backend variants of the Mozannar--Sontag and one-vs-all L2D surrogates~\cite{mozannar2020,verma2022}; they are not full reproductions of joint classifier--rejector training. Their class logits remain fixed at $(\log(1-\hat p),\log\hat p)$, while a two-layer defer head uses the signal bank. The published losses are optimized on 707 temporal out-of-fold rows with simulated 90\%-accurate expert decisions; calibration cost selects thresholds and B6's $\alpha\in\{0,.5,1\}$. We report three seeds.

% -------------------------------------------------------
\section{Main Comparison}
\label{sec:main}

\begin{figure}[t]
\centering
\includegraphics[width=\columnwidth]{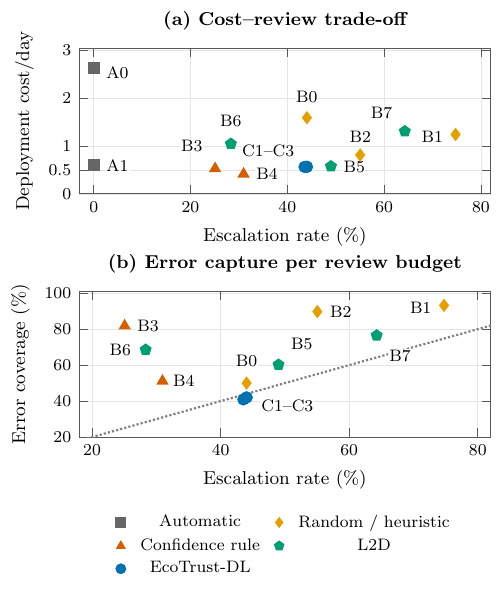}
\caption{Chow provides the best in-distribution cost--review trade-off, while high error coverage often requires broad review. (a) Deployment cost versus escalation rate; lower-left is preferable. (b) Error coverage versus escalation rate; upper-left is preferable and the dotted diagonal is random ranking. Marker families separate automatic policies, random/earlier heuristics, confidence rules, L2D, and EcoTrust-DL; codes A0--B7 are defined in the baseline paragraph and C1--C3 in the meta-risk-estimator paragraph (Sec.~\ref{sec:setup}), and B6--B7 are three-seed means. All policies use the same frozen predictor ($c_{FN}{:}c_{FP}{:}c_h{=}100{:}3{:}1$, test $n{=}1{,}208$).}
\label{fig:main}
\end{figure}

Figure~\ref{fig:main} gives the central comparison. First, the earlier heuristics do not justify their tuning. B1 escalates 74.8\% of days and costs 1.244, while B2 reviews 55.1\% and costs 0.814. Both are worse than the fully automatic cost-sensitive threshold (0.596).

Second, \textbf{Chow is the strongest in-distribution policy}: it costs 0.416 per day versus EcoTrust's 0.567. At $c_{FN}{=}100$, EcoTrust must estimate risk near the predicted-negative threshold $c_h/c_{FN}=0.01$ from 1{,}195 meta-training rows. It reviews 44.0\% of days, compared with Chow's 31.0\%. The two post-hoc L2D variants cost 1.048 and 1.310. For a well-calibrated in-distribution backend, these results favor Chow.

Third, the learned estimator remains the best error model by a small margin. Its AUC/Brier are 0.889/0.056, compared with 0.869/0.058 for confidence and 0.861/0.061 for the raw-feature rejector. Reliability is monotone by decile. Better error ranking therefore does not by itself imply lower decision cost.

B1 has the highest error coverage (93.2\%) but reviews nearly three-quarters of all cases. This contrast illustrates why error coverage alone is insufficient for a deployment that charges for review.

\subsection{Which Signals Carry the Risk?}

All 63 non-empty subsets were evaluated; Fig.~\ref{fig:mechanism}(b)--(c) shows the leave-one-out slice. No signal is individually load-bearing: removing one changes cost by at most 0.011 and AUC by at most 0.015. In fact, no removal increases cost, while entropy and conformal information contribute most to AUC. We retain the pre-specified full bank rather than select features using test cost.

% -------------------------------------------------------
\section{When Does a Learned Risk Model Pay?}
\label{sec:calib}

Section~\ref{sec:main} raises a practical question: if confidence already yields a cheaper policy, when is a separate risk estimator useful? Calibration is one plausible moderator because a confidence rule relies directly on the backend probability scale, whereas the meta-model can correct systematic distortions.

We test this mechanistically with 12 off-the-shelf backends and the primary ensemble (13 real backends), plus one deliberately over-confident control. Randomized models use three seeds; confidence intervals resample complete test years, preserving temporal dependence.

\begin{figure*}[t]
\centering
\includegraphics[width=0.93\textwidth]{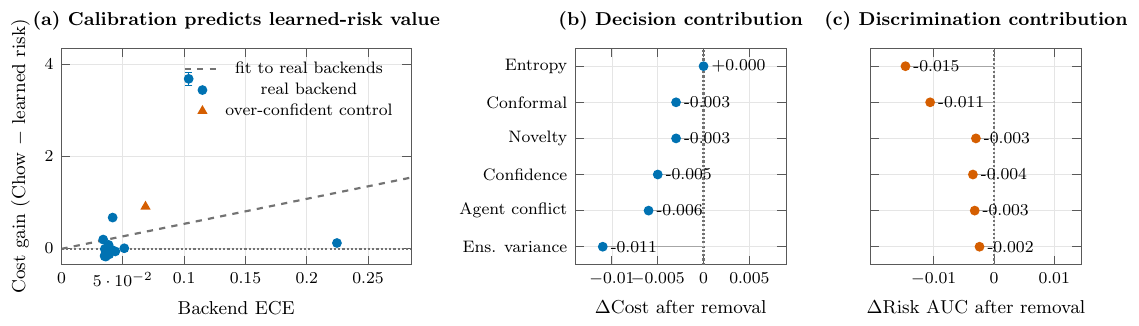}
\caption{Calibration partly predicts learned-risk value, while no signal is individually load-bearing. (a) Backend expected calibration error (ECE) versus cost gain (Chow cost minus EcoTrust-DL cost); the dashed fit excludes the over-confident control. Error bars show across-seed s.d. (b)--(c) Removing one signal changes test cost by at most 0.011 and risk AUC by at most 0.015. Positive $\Delta$Cost means the removed signal helped the decision; negative $\Delta$AUC means it helped discrimination.}
\label{fig:mechanism}
\end{figure*}

Excluding the control, Pearson correlation is 0.278 with year-block bootstrap 95\% CI [0.120, 0.405]; Spearman correlation is 0.432 [--0.049, 0.736]. This bootstrap resamples test years, so the interval is temporal uncertainty \emph{conditional on this fixed backend bank}, not uncertainty over backends; we read 0.278 as descriptive within the bank, not an inferential guarantee. Learned risk wins on six of 13 real backends: XGBoost, logistic regression, Gaussian Naive Bayes, two decision trees, and AdaBoost. For the ensemble it costs 0.567, compared with Chow's 0.416, while it reduces the over-confident control from 1.478 to 0.566. The positive Pearson value is consistent with an association, but the rank interval includes zero. Miscalibration is therefore an informative diagnostic, not a sufficient condition.

\textbf{Is the value merely recalibration?} Because $g_\phi$ receives supervised error labels while a confidence rule uses the raw probability, we add a calibration-matched, class-aware recalibrator. It uses the same Platt/logistic machinery, rows, folds, and support gate, but only $[\,|\hat p-\tfrac12|,\ \hat y,\ |\hat p-\tfrac12|\!\cdot\!\hat y\,]$. Table~\ref{tab:recalibration} shows that the full bank is cheaper on seven of 12 backends, but its mean cost is slightly higher and the paired year-block interval crosses zero. Its advantage over raw confidence is therefore largely consistent with supervised, class-aware recalibration rather than a reliably beneficial six-signal mechanism.

\begin{table}[t]
\centering
\caption{Calibration-matched comparison over 12 backends. Costs are averaged over three seeds; $\Delta$ is full-bank minus class-aware cost per day.}
\label{tab:recalibration}
\fittab{
\begin{tabular}{@{}lr@{}}
\toprule
Comparison or statistic & Result \\
\midrule
Full bank cheaper than raw Chow & 6/12 \\
Class-aware cheaper than raw Chow & 6/12 \\
Full bank cheaper than class-aware & 7/12 \\
Mean $\Delta$ (year-block 95\% CI) & $+0.0082$ [$-0.0002$, $0.0172$] \\
Across-backend $\Delta$ range & [$-0.017$, $+0.109$] \\
Pearson(ECE, $\Delta$) & 0.011 \\
\bottomrule
\end{tabular}}
\tabnote{Positive $\Delta$ favors the class-aware recalibrator. The interval uses 2{,}000 paired year-block resamples. All methods use the primary cost setting and a perfect simulated reviewer.}
\end{table}

% -------------------------------------------------------
\section{Cost Sensitivity and Review Budgets}
\label{sec:costsweep}

\begin{figure*}[t]
\centering
\includegraphics[width=0.82\textwidth]{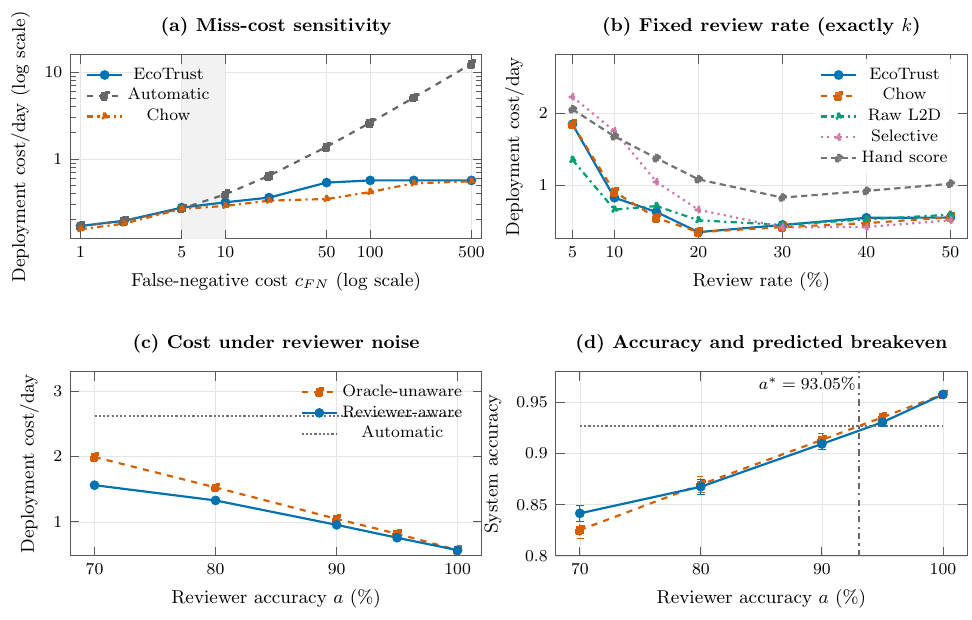}
\caption{Deployment costs depend on error severity, review budget, and reviewer quality. (a) EcoTrust becomes cheaper than automation between $c_{FN}=5$ and $10$, although Chow remains cheaper in-distribution. (b) No ranking dominates at exact review rates. (c)--(d) Reviewer-aware decisions reduce cost under noise; the predicted breakeven $a^*=93.05\%$ is near the accuracy crossing. Accuracy bars show Monte Carlo s.d.}
\label{fig:operating}
\end{figure*}

Figure~\ref{fig:operating}(a) sweeps the miss cost across three orders of magnitude. Escalation rises from 1.7\% and plateaus at 44.0\%, while $c_h/c_{FN}$ falls from 1 to 0.002; no policy threshold is searched.

When misses are cheap, review can cost more than the errors it prevents: at $c_{FN}=5$, EcoTrust costs 0.275 versus 0.268 for automation. At $c_{FN}=10$, it is 19.2\% cheaper, and at $c_{FN}=100$ it is 78.4\% cheaper. Chow still costs less at the primary operating point (0.416 vs.\ 0.567).

Figure~\ref{fig:operating}(b) ranks cases by expected saving under each method's own risk estimate. The hand-weighted composite is dominated at every displayed review rate. Among principled scores, the raw-feature rejector is best at 5--10\%, Chow at 15--20\%, and selective classification at 30--50\%. The non-monotone curves are expected because the comparison forces exactly $k$ reviews; the oracle at-most-$k$ diagnostic is non-increasing and convex as Prop.~\ref{prop:budget} predicts.

% -------------------------------------------------------
\section{Evaluation Without the Oracle}
\label{sec:reviewer}

Many deferral evaluations assume a perfect reviewer. Figure~\ref{fig:operating}(c)--(d) relaxes that assumption. A simulated reviewer flips each deferred decision with probability $1-a$; results average 200 Monte Carlo draws with standard deviations. The \emph{oracle-unaware} policy is optimized for $a{=}1$, whereas the \emph{reviewer-aware} policy uses the specified $a$ in rule~\eqref{eq:rule}.

Prop.~\ref{prop:breakeven} predicts $a^\star=0.9305$, the model's accuracy on EcoTrust's escalated set. The simulation straddles it: oracle-unaware deferral has accuracy 0.913 at $a=0.90$ and 0.936 at $a=0.95$, compared with 0.927 without review.

Reviewer-awareness reduces cost but cannot compensate for a poor reviewer. At $a=0.70$, it lowers cost from 1.994 to 1.563 and escalation from 44.0\% to 36.8\%, while accuracy remains only 0.842. The same 70--100\% sweep covers Chow, selective classification, the raw-feature rejector, and both post-hoc L2D variants. Reviewer accuracy is therefore a policy input to be estimated from audited decisions, not an oracle assumption.

% -------------------------------------------------------
\section{Detected Extrapolation: A Conservative Fail-Safe}
\label{sec:ood}

\begin{figure*}[t]
\centering
\includegraphics[width=0.84\textwidth]{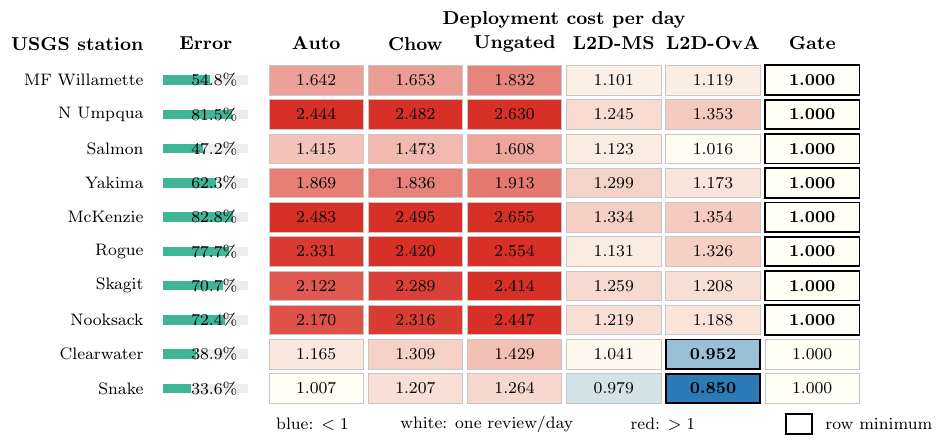}
\caption{Under station transfer, the support gate is identical to an always-review policy and costs one per day with a perfect reviewer. It is cheapest on eight stations; L2D-OvA is cheapest on Clearwater and Snake. No target labels or tuning are used. Blue cells cost less than one review per day, red cells cost more, and outlines mark row minima.}
\label{fig:transfer}
\end{figure*}

Transfer error rates range from 33.6\% to 82.8\% (Fig.~\ref{fig:transfer}). The default Mahalanobis gate marks every day at every station out-of-support. With a perfect reviewer, it therefore becomes exactly the always-review baseline: 100\% escalation at cost $c_h=1$. This is cheapest on eight sites; learning-to-defer one-vs-all (L2D-OvA) is cheaper on Clearwater and Snake. The result verifies a deterministic conservative fallback, not transfer competence. A system requiring 100\% review needs local calibration before sustained deployment.

Chow escalates 8.2--35.6\% but is never cheapest. The post-hoc L2D variants escalate 68.2--80.7\% and cost 0.850--1.354; OvA wins two sites. The gate's distinction is therefore not unique shift sensitivity, but behavior specified before target labels are observed.

Fixed review quotas remove the full-review guarantee. We cross reviewer accuracy $a\in\{0.80,0.90,0.95\}$ with forced review rates $B\in\{0.10,0.30,0.50\}$ and average cost across stations. Here $B$ is a \emph{forced} review rate: every policy reviews exactly $B$ of cases, including negative-saving ones. This is a workload quota, not an at-most-$B$ constraint; by Prop.~\ref{prop:budget} an at-most-$B$ policy would weakly dominate, so these are conservative upper bounds, and the comparison is fair because all policies are forced identically. EcoTrust and Chow are nearly indistinguishable under shift. At $a\leq0.90$, neither beats the best competing policy in any budget setting. At $a=0.95,B=0.50$, EcoTrust costs 1.773 versus 1.865 for automation and 1.778 for Chow; at $B=0.30$, OvA is best at 1.694. Thus the gate defines a fallback under detected extrapolation, but its ranking does not recover transferable case-level risk.

Empirical- and Ledoit--Wolf-Mahalanobis gates flag all ten stations for every $\epsilon\in[0.01,0.20]$. Standardized Euclidean distance does so for $\epsilon\ge0.025$ (three stations at 0.01). Across the empirical-Mahalanobis sweep, in-distribution flagging rises from 1.2\% to 16.6\% and cost from 0.565 to 0.593, exposing the review-load trade-off.

% -------------------------------------------------------
\section{Discussion}
\label{sec:discussion}

\textbf{What the evidence supports.} The framework connects asymmetric costs, reviewer quality, and budgeted review in a single deployable comparison. Consistent with prior analyses showing that confidence sufficiency depends on the model and deployment setting~\cite{jitkrittum2023confidence,cattelan2024broken}, the empirical recommendation is conditional: use Chow for the well-calibrated in-distribution backend studied here, and consider a learned risk model only after validating lower cost on temporal calibration data. The measured accuracy crossing is close to the predicted reviewer breakeven, and the perfect-reviewer fallback bound is attained at all ten transfer stations.

\textbf{What it does not support.} Learned risk is not uniformly better than confidence: it costs 0.567 versus Chow's 0.416 for the main backend and wins on six of 13 alternatives. The positive Pearson association with ECE is moderate, while the Spearman interval includes zero. Under transfer, post-hoc L2D can approach or beat the gate, and at fixed forced review rates no shifted-domain policy dominates. EcoTrust therefore characterizes deployment choices rather than establishing algorithmic dominance or general shift robustness.

\textbf{Seasonal structure.} A month-only logistic regression reaches 94.1\% accuracy in August--November, above the ensemble's 92.7\%. The study does not claim to improve raw accuracy. Its target is the downstream deferral decision under costs and review constraints, while the strong calendar baseline confirms that this remains a comparatively simple classification task.

\textbf{Limitations.} The meta risk model uses 1{,}195 rows, including 707 temporal out-of-fold cases; a longer record is needed to resolve the $0.01$ operating tail. The ten stations broaden geography but not task domain, and several have no positive August--November days under the fixed $18^\circ$C label. The gate detects covariate rather than concept shift and reviews every transfer case. Reviewer errors and L2D expert decisions are simulated and class-independent; actual difficulty-dependent behavior may differ. Finally, B6 and B7 freeze class logits for comparability and are not full joint-training reproductions.

% -------------------------------------------------------
\section{Conclusion}
\label{sec:conclusion}

We have presented a controlled case study of confidence rejection and learned-risk deferral around a frozen predictor. Temporal evaluation shows that better error ranking need not lower decision cost: Chow is cheaper for the primary in-distribution backend, while learned risk wins on six of 13 alternatives. A calibration-matched test does not resolve a mean cost difference between the full signal bank and a class-aware recalibrator. Under station transfer, the support gate becomes an always-review fallback rather than demonstrating transferable risk estimation. EcoTrust's contribution is therefore an auditable decision protocol and a set of measured operating boundaries, not a claim of algorithmic dominance. Validation on additional tasks and with real reviewers remains necessary.

% -------------------------------------------------------
\bibliographystyle{IEEEtran}
\bibliography{IEEEabrv,references}

\end{document}